\documentclass[11pt,a4paper]{article}
\usepackage{amsmath}
\usepackage[T1]{fontenc}
\usepackage[utf8]{inputenc}
\usepackage{amsthm}
\usepackage{amstext}
\usepackage{amssymb}
\usepackage{stmaryrd}
\usepackage{mathrsfs}
\usepackage{bbold}
\usepackage{complexity}
\usepackage{graphicx}
\usepackage{hyperref}
\usepackage{url}
\usepackage{verbatim}
\usepackage{multicol}
\usepackage{authblk}
\usepackage{color}

\newtheorem{theorem}{Theorem}

\newtheorem{lemma}[theorem]{Lemma}
\newtheorem*{lemma*}{Lemma}
\newtheorem{proposition}[theorem]{Proposition}
\newtheorem{corollary}[theorem]{Corollary}

\newtheorem*{open*}{Open~question}
\newtheorem{definition}[theorem]{Definition}
\newcommand{\e}{{\mathrm e}}
\begin{document}

\title{A Wronskian Approach to the Real $\tau$-Conjecture}

\author{Pascal Koiran, Natacha Portier and S\'ebastien Tavenas}

\affil{LIP\footnote{UMR 5668 ENS Lyon - CNRS - UCBL - INRIA,
    Université de Lyon, \{pascal.koiran,natacha.portier,sebastien.tavenas\}@ens-lyon.fr}, École Normale Supérieure de Lyon}

\date{\today}

\maketitle

\begin{abstract}
  According to the real $\tau$-conjecture, the number of real roots of
  a sum of products of sparse univariate polynomials should be
  polynomially bounded in the size of such an expression. It is known
  that this conjecture implies a superpolynomial lower bound on the
  arithmetic circuit complexity of the permanent.
 
  In this paper, we use the Wronksian determinant to give an upper
  bound on the number of real roots of sums of products of sparse
  polynomials of a special form. We focus on the case where the number
  of distinct sparse polynomials is small, but each polynomial may be
  repeated several times. We also give a deterministic polynomial
  identity testing algorithm for the same class of polynomials.

  Our proof techniques are quite versatile; they can in particular be
  applied to some sparse geometric problems that do not originate from
  arithmetic circuit complexity.  The paper should therefore be of
  interest to researchers from these two communities (complexity
  theory and sparse polynomial systems).
\end{abstract}

%\thispagestyle{empty}

%\newpage

\section{Introduction}

\setcounter{page}{1}

The complexity of the permanent polynomial 
$${\mathrm{per}}(x_{11},\ldots,x_{nn})=\sum_{\sigma \in S_n} 
\prod_{i=1}^n x_{i\sigma(i)}$$ is one of the central open problems in
complexity theory.  It is widely believed that the permanent is not
computable by arithmetic circuits of size polynomial in~$n$. This
problem can be viewed as an algebraic version of the P versus NP
problem~\cite{Valiant82,Burg}.

It is known that this much coveted lower bound for the permanent would
follow from a so-called {\em real $\tau$-conjecture} for sums of
products of sparse univariate polynomials~\cite{Koi10}.  Those are
polynomials in $\mathbb{R}[x]$ of the form $\sum_{i=1}^k \prod_{j=1}^m
f_{ij}(x)$, where the sparse polynomials $f_{ij}$ have at most $t$
monomials.  According to the real $\tau$-conjecture, the number of
real roots of such an expression should be polynomially bounded in
$k$, $m$ and $t$.  The original $\tau$-conjecture by Shub and
Smale~\cite{ShSm95} deals with integer roots of arbitrary
(constant-free) straight-line programs.

As a first step toward the real $\tau$-conjecture, Grenet, Koiran,
Portier and Strozecki~\cite{GKPS11} considered the family of sums of
products of powers of sparse polynomials.  Those polynomials are of
the form
\begin{equation}\label{spspower} 
  \sum_{i=1}^k \prod_{j=1}^m f_{j}^{\alpha_{i,j}}.
\end{equation}
They are best viewed as sums of products of sparse polynomials where
the total number $m$ of distinct sparse polynomials is ``small'', but
each polynomial may be repeated several times. In particular, if one
can find a $(kt)^{O(1)}2^{O(m)}$ upper bound on the number of real
roots, then it will imply the real-$\tau$ conjecture in the case where
the number of distinct sparse polynomials is logarithmically bounded.
The upper bound on the number of real roots obtained in~\cite{GKPS11}
is polynomial in $t$, but exponential in $m$ and doubly exponential in
$k$.

Bounds on the number of real zeros for systems of sparse polynomials
were extensively studied by Khovanski\u{\i}~\cite{Kh91} in his
``fewnomial theory''. His results imply an upper bound exponential in
$k$, $m$ and $t$. In this article, we will give a bound of order
$t^{O(k^2m)}$, thereby removing the double exponential
from~\cite{GKPS11} while staying polynomial in $t$. Moreover, our
results extend well to some other families of functions. In
particular, they extend a result from Avenda\~no~\cite{Ave09} on the
intersection of a sparse plane curve and a line. He gave a linear
bound on the number of roots for polynomials of the form $\sum_{i=1}^k
x^{\alpha_i}(ax+b)^{\beta_i}$ where $\alpha_i$ and $\beta_i$ are
integers and gave an example proving that his linear bound does not
apply for non-integer powers. Our result gives a polynomial upper
bound for the wider family~(\ref{spspower}) where the polynomials
$f_j$ are of bounded degrees and the $\alpha_{i,j}$ are real
exponents.

In addition to bounds on the number of real roots, we also give a
deterministic identity testing algorithm for polynomials of the
form~(\ref{spspower}). The running time of our algorithm is polynomial
in $t$, in the bit size of coefficients and of the powers
$(\alpha_{i,j})$ and exponential in $k$ and $m$.  Polynomial Identity
Testing (PIT) is a very well-studied problem.  The Schwartz-Zippel
lemma yields a randomized algorithm for PIT, but the existence of an
efficient deterministic algorithm is an outstanding open
problem. Connections between circuit lower bounds and deterministic
PIT algorithms were discovered in 1980 by Heintz and
Schnorr~\cite{HS80}, then more recently by Kabanets and
Impagliazzo~\cite{KI04}, by Aaronson and van Melkebeek~\cite{AM11} and
by Agrawal~\cite{Agra05}. Recently, many deterministic PIT algorithms
have been found for several restricted models (see e.g. the two
surveys~\cite{AS09,Sax09}). In particular, a deterministic PIT
algorithm for functions of the form of~(\ref{spspower}) has already
been given in~\cite{GKPS11}. Their algorithm is polynomial in $t$,
exponential in $m$ but doubly exponential in $k$ whereas we give a new
algorithm which is only exponential in $k$.

We now present our main technical tools. Finding the roots of a
product of polynomials is easy: it is the union of the roots of the
corresponding polynomials. But finding the roots of a sum is
difficult: for example how to bound the number of real roots of $fg+1$
where $f$ and $g$ are $t$-sparse? It is an open question to decide if
this bound is linear in $t$. Our main tool in this paper to tackle the
sum is the Wronskian. We recall that the Wronskian of a family of
functions $f_1,\ldots,f_k$ is the determinant of the matrix of their
derivatives of order 0 up to $k-1$. More formally,
$$ W(f_1,\ldots,f_k) = \det \left(\left( f_j^{(i-1)}\right)_{1\leq i,j \leq k}\right). $$
The Wronskian is useful especially for its connection to linear
independence (more on this in the next section).  Another classical
and very useful tool is Descartes' rule of signs:
\begin{lemma}[Strong rule of signs]\label{lem_strongDescartes}
  Let $f=\sum_{i=1}^t a_ix^{\alpha_i}$ be a polynomial such that
  $\alpha_1 < \alpha_2 < \ldots < \alpha_t$ and $a_i$ are nonzero real
  numbers. Let $N$ be the number of sign changes in the sequence
  $(a_1,\ldots,a_t)$. Then the number of positive real roots of $f$ is
  bounded by $N$.
\end{lemma}
In this article, we will use a weak form of this lemma.
\begin{lemma}[Weak rule of signs]\label{lem_Descartes}
  Let $f=\sum_{i=1}^t a_ix^{\alpha_i}$ be a polynomial such that
  $\alpha_1 < \alpha_2 < \ldots < \alpha_t$ and $a_i$ are nonzero real
  numbers. Then the number of positive real roots of $f$ is bounded by
  $t-1$. Moreover, the result is also true in the case where the
  exponents $\alpha_i$ are real.
\end{lemma}
In their book \cite{PoSz76}, P\'olya and Szeg\H{o} gave a
generalization of the strong rule of sign using the Wronskian. Some
relations were already known between Wronskians and the sparse
polynomials (c.f. for example~\cite{GKS91},~\cite{GK93}
and~\cite{GKS94}). We show in Theorem~\ref{Thm_heart} that bounding
the number of roots of the Wronskian yields a bound on the number of
roots of the corresponding sum. In general, the Wronskian may seem
more complicated than the sum of the functions, but for the families
studied in this paper it can be factorized more easily
(Theorems~\ref{Thm_model1} and~\ref{Thm_model2}).

The paper is organized as follows. The main results of
Section~\ref{Sec_upperbound} are Theorems~\ref{Thm_heart},
\ref{thm_main} and~\ref{thm_main3}, which bound the number of roots of
sums as a function of the number of roots of the Wronskian.  Then, in
Section~\ref{Sec_applications}, we apply these results to particular
families of polynomials. The main applications that we have in mind
are to polynomials of the form~(\ref{spspower}), and to the
polynomials studied by Avenda\~no.  We give in Section~\ref{Sec_PIT}
some PIT algorithms for polynomials of the form~(\ref{spspower}). The
proof of Theorem~\ref{thm_main3} will be given in
Section~\ref{Sec_upperbound_refinement}.  And finally, we show in
Section~\ref{Sec_optimality} that our method is optimal in a precise
sense. Some of the proofs are postponed to the Appendix.

\section{Zeros of the Wronskian as an upper
  bound} \label{Sec_upperbound}

Let us recall that for a finite family of real functions
$f_1,\ldots,f_k$ sufficiently differentiable, the Wronskian is defined
by
\begin{align*}
  W(f_1,\ldots,f_k) = \det \left( \left( f_j^{(i-1)}\right)_{1 \leq i,j \leq k}\right).
\end{align*}

We will use the following properties of the Wronskian.
\begin{lemma}
  Let $f_1,\ldots,f_k$ and $g$ be $k-1$ times
  differentiable real functions.
  
  Then, $W(gf_1,\ldots,gf_k)=g^kW(f_1,\ldots,f_k)$.
\end{lemma}

As a corollary:
\begin{lemma}
  Let $f_1,\ldots,f_k$ be $k-1$ times differentiable real functions 
  and let $I$ be an interval where they do not vanish.
  
  Then, over $I$, we have $W(f_1, \ldots , f_k)=\left(f_1\right)^k
  W\left(\left(\frac{f_2}{f_1}\right)^\prime, \ldots ,
    \left(\frac{f_k}{f_1}\right)^\prime\right).$
\end{lemma}

These results can be found in \cite{PoSz76} (ex. 57, 58 in Part 7).
Notice that the Wronskian of a linearly dependent family of functions
is identically zero (if a family is dependent then the family of the
derivatives is also dependent with the same coefficients).  But the
converse is not necessarily true. Peano, then B\^ocher, found
counterexamples \cite{Pea1889a,Pea1889b,Bo1900a} (see \cite{HistWr}
for a history of these results). However, B\^ocher \cite{Bo1900b}
proved that this converse becomes true if the functions are analytic
\cite{Hur90}.

\begin{lemma}\label{lem_zero} 
  If $f_1,\ldots,f_k$ are analytic functions, then $\left(f_i\right)$
  is linearly dependent if and only if $W(f_1,\ldots,f_k)=0$.
\end{lemma}

\begin{definition}
  For every function $g$ and interval $I$, we will denote $Z_I(g)$ the
  number of distinct real roots of $g$ over $I$. We just write $Z(g)$
  when the interval is clear from the context.
\end{definition}

Throughout the paper we consider only intervals that are not reduced
to a single point (and we allow unbounded intervals).  The next
theorem is in fact implied, in the analytic case, by Voorhoeve and Van
Der Poorten's result~\cite{VP75} (see below, Theorem~\ref{Thm_Voor}
and the following paragraph for more precisions). In
Theorem~\ref{Thm_heart} we only assume that the $f$ are sufficiently
differentiable.

\begin{theorem}\label{Thm_heart}
  Let $k$ be a non zero integer. Let $f_i$ be $k$ functions $k-1$
  times differentiable in an interval $I$ such that for all $i \leq
  k$, the Wronskian $W(f_1, \ldots, f_i)$ does not have any zero over
  $I$.

  If the real constants $a_1,\ldots,a_k$ are not all equal to 0,
  $a_1f_1+a_2f_2+\ldots + a_kf_k$ has at most $k-1$ real zeros over
  $I$ counted with multiplicity.
\end{theorem}

\begin{proof}
  We show this result by induction on $k$.  If $k=1$, then,
  $f_1=W(f_1)$ does not have any zero. Moreover, $a_1$ is not
  zero. So, $a_1f_1$ has no zeros.

  For some $k\geq 2$, let us suppose that the property is true for all
  linear combinations of size $k-1$. Denote $z$ the number of zeros of
  $a_1f_1+\ldots + a_kf_k$. If $a_2=a_3=\ldots=a_k=0$, then $a_1\neq
  0$ and $a_1f_1+a_2f_2+\ldots + a_kf_k=a_1f_1$ has no zero, and the
  conclusion of the theorem holds true.  Otherwise,
  $a_1+\frac{a_2f_2}{f_1}+ \ldots + \frac{a_kf_k}{f_1}$ has $z$ zeros
  (since $f_1=W(f_1)$ does not have any zero by hypothesis). By
  application of Rolle's Theorem,
  $a_2\left(\frac{f_2}{f_1}\right)^\prime+\ldots+a_k\left(\frac{f_k}{f_1}\right)^\prime$
  has at least $z-1$ zeros over $I$.

  Function $f_1$ does not have any root in $I$, so the functions
  $\left(\frac{f_2}{f_1}\right)^\prime,\ldots,\left(\frac{f_k}{f_1}\right)^\prime$
  are $k-2$ times differentiable. Moreover, for all $2 \leq i \leq k$,
  $W \left(
    \left(f_2/f_1\right)^\prime,\ldots,\left(f_{i}/f_1\right)^\prime\right)
  = W(f_1, \ldots, f_i)/f_1^i$ does not have any roots over $I$.
  Since the coefficients $a_2,\ldots,a_k$ are not all zero, by
  induction hypothesis $a_2\left(\frac{f_2}{f_1}\right)^\prime+ \dots
  + a_k\left(\frac{f_k}{f_1}\right)^\prime$ has at most $k-2$
  zeros. Hence $ a_1f_1+a_2f_2+ \ldots+ a_kf_k $ has at most $k-1$
  zeros by Rolle's theorem.
\end{proof}

The following theorem gives us a method to find upper bounds on the
number of roots. We will show in Section~\ref{Sec_optimality} that it
is sometimes tight.

\begin{theorem}\label{thm_main}
  Let $f_1, \ldots, f_k$ be analytic functions on an interval $I$.  If
  the real constants $a_1,\ldots,a_k$ are not all equal to 0,
  \begin{align}\label{Eq1} 
    Z(a_1f_1+\ldots + a_kf_k) \leq \left(1+
      \underset{i=1}{\overset{k}{\sum}}Z(W(f_1, \ldots,
      f_i))\right)k-1.
  \end{align}
  More precisely, if $\Upsilon = \left\{x\in I | \exists i \leq k,
    W(f_1, \ldots, f_i)(x)=0 \right\}$ is finite, then
  \begin{align}\label{Eq2} 
    Z(a_1f_1+\ldots + a_kf_k) \leq \left(1+ |\Upsilon|\right)k-1.
  \end{align}
  Moreover, the inequalities still hold if on the left side, zeros
  which are not zero of one of the Wronskians $W(f_1, \ldots, f_i)$
  are counted with multiplicity.
\end{theorem}

\begin{proof} 
  We will directly prove the more precise
    version~(\ref{Eq2}).
  If $a_1f_1+\ldots+a_kf_k$ is the zero polynomial, then the family is
  linearly dependent and so the Wronskian $W(f_1,\ldots,f_k)$ is also
  the zero polynomial.  This means that $\Upsilon = I$ is infinite and
  the inequality is verified.

  Otherwise, $a_1f_1+\ldots+a_kf_k$ has a finite number of zeros. We
  have $\Upsilon = \underset{i=1}{\overset{k}{\bigcup}}Z(W(f_1,
  \ldots, f_i)) $. So, $|\Upsilon| \leq
  \underset{i=1}{\overset{k}{\sum}}|Z(W(f_1, \ldots, f_i))|$ and we
  will prove~(\ref{Eq2}). The set $I \setminus \Upsilon$ is an union
  of $|\Upsilon| + 1$ intervals. Let $J$ be one of these
  intervals. With Theorem~\ref{Thm_heart}, we get
  $Z_J(a_1f_1+\ldots+a_kf_k) \leq k-1$. So $a_1f_1+\ldots+a_kf_k$ has
  at most $\left(1+ |\Upsilon|\right)(k-1)$ zeros over $I \setminus
  \Upsilon$ and at most $\left(1+ |\Upsilon|\right)(k-1)+|\Upsilon|$
  zeros over $I$.
\end{proof}

In Section~\ref{Sec_upperbound_refinement}, we will prove the
following variation on Theorem~\ref{thm_main}:

\begin{theorem}\label{thm_main3} 
  Let $f_1, \ldots, f_k$ be analytic linearly independent functions on
  an interval $I$. Then,
  \begin{align*}
    Z(f_1+\ldots + f_k) \leq k-1 +
    Z(W_k)+Z(W_{k-1})+2\sum_{j=1}^{k-2}Z(W_j).
  \end{align*}
\end{theorem}
In most applications, this result yields a better bound than
Theorems~\ref{Thm_heart} and~\ref{thm_main}.

\section{Applications}\label{Sec_applications}

In this section, we prove Theorem~\ref{Thm_model1} which bounds the
number of zeros of the polynomials of the form (\ref{spspower}). The
given bound improves both Grenet, Koiran, Portier and Strozecki's
result~\cite{GKPS11} and the bound implied by Khovanski\u{\i}'s
fewnomial theory~\cite{Kh91}.  At the end of this section, we also
extend Avenda\~no's result to real exponents.  We saw before (in
Section~\ref{Sec_upperbound}) that the number of zeros of a linear
combination of real functions can be bounded as a function of the
number of zeros of their Wronskians.  As a result, it remains to bound
the number of zeros of Wronskians of polynomials of the form
$\prod_{j=1}^m f_{j}^{\alpha_{i,j}}$.  Such a Wronskian has few zeros
thanks to a nice factorization property: after factoring out some high
powers, we are left with a determinant whose entries are low-degree
polynomials (or sparse polynomials, depending on the model under
consideration).  It is then straightforward to bound the number of
real roots of this determinant.

\subsection{Derivative of a power}

We use ultimately vanishing sequences of integer numbers, i.e.,
infinite sequences of integers which have only finitely many nonzero
elements. We denote this set $\mathbb{N}^{(\mathbb{N})}$.  For any
positive integer $p$, let $\mathscr{S}_p = \{ (s_1,s_2,\ldots) \in
\mathbb{N}^{(\mathbb{N})} | \ \underset{i=1}{\overset{\infty}{\sum}} i
s_i = p \}$ (so for each $p$, this set is finite). Then if $s$ is in
$\mathscr{S}_p$, we observe that for all $i\geq p+1$, we have
$s_i=0$. Moreover for any $p$ and any $s=(s_1,s_2,\ldots) \in
\mathbb{N}^{(\mathbb{N})}$, we will denote $|s| =
\underset{i=1}{\overset{\infty}{\sum}} s_i$ (the sum makes sense
because it is finite).

\begin{lemma} \label{lem_power} Let $p$ be a positive integer and
  $\alpha \geq p$ be a real number.  Then \[\left(f^\alpha
  \right)^{(p)} = \sum_{s \in \mathscr{S}_{p}} \left[ \beta_{\alpha,s}
    f^{\alpha-|s|} \prod_{k=1}^p \left( f^{(k)} \right)^{s_k}
  \right] \] where $(\beta_{\alpha,s})$ are some constants.
\end{lemma}

The order of differentiation of a monomial $\prod_{k=1}^p
(f^{(k)})^{s_k}$ is $\sum_{k=1}^p ks_k$. The order of differentiation of
a differential polynomial is the maximal order of its monomials. For
example: if $f$ is a function, the total order of differentiation of
$f^3\left(f^\prime\right)^2\left(f^{(4)}\right)^3+3ff^\prime$ is $\max
(3*0+2*1+3*4,0*1+1*1)=14$.

Lemma~\ref{lem_power} just means that the $p$-th derivative of a power
$\alpha$ of a function $f$ is a linear combination of terms such that
each term is a product of derivatives of $f$ of total degree $\alpha$
and of total order of differentiation $p$. 

\begin{proof}
  In the following, $e_i$ is the sequence
  $(0,0,\ldots,0,1,0,0,\ldots)$ where the $1$ appears at the
  $i^{\rm{th}}$ coordinate.  We show this lemma by induction over $p$.
  If $p=1$, then $\left( f^\alpha \right)^\prime = \alpha f^\prime
  f^{\alpha -1}$. That is the basis case since
  $\mathscr{S}_1=\{(1,0,0,\ldots)\}$. We notice that
  $\beta_{\alpha,(1,0,\ldots)}=\alpha$.

  Let us suppose that the lemma is true for a fixed $p$. By induction
  hypothesis, we have
  \begin{align*}
    \left( f^\alpha \right)^{(p+1)}
    & = \left(\sum_{s \in \mathscr{S}_{p}} \beta_{\alpha,s} f^{\alpha-|s|} \prod_{k=1}^p \left( f^{(k)} \right)^{s_k} \right)^\prime \\
    & = g_1 + g_2
  \end{align*}
  where
  \begin{align*}
    g_1= \sum_{s \in \mathscr{S}_{p}} \beta_{\alpha,s} \left(f^{\alpha-|s|} \right)^\prime \left( \prod_{k=1}^p \left( f^{(k)} \right)^{s_k} \right) \\
    g_2= \sum_{s \in \mathscr{S}_{p}} \beta_{\alpha,s} f^{\alpha-|s|}
    \left(\prod_{k=1}^p \left( f^{(k)} \right)^{s_k} \right)^\prime.
  \end{align*}
  By rewriting each term, we get
  \begin{align*}
    g_1 & = \sum_{s \in \mathscr{S}_{p}} \beta_{\alpha,s} (\alpha-|s|) f^\prime f^{\alpha-|s|-1} \prod_{k=1}^p \left( f^{(k)} \right)^{s_k} \\
    & = \sum_{\underset{s^\prime = s+e_1}{s \in \mathscr{S}_{p}}}
    \beta_{\alpha,s} (\alpha-|s^\prime|+1) f^{\alpha-|s^\prime|}
    \prod_{k=1}^p \left( f^{(k)} \right)^{s^\prime_k}
  \end{align*}
  \begin{align*}
    g_2 & =  \sum_{s \in \mathscr{S}_{p}} \beta_{\alpha,s}
    f^{\alpha-|s|} \sum_{j=1}^p s_j f^{(j+1)} \left( f^{(j)}
    \right)^{s_j-1} 
    \prod_{k \neq j} \left( f^{(k)} \right)^{s_k} \\
    & = \sum_{j=1}^p \left[ \sum_{\underset{s^\prime =
          s-e_j+e_{j+1}}{\underset {\: s_j \neq 0}{s \in
            \mathscr{S}_{p}}}} \beta_{\alpha,s^\prime+e_j-e_{j+1}}
      f^{\alpha-|s^{\prime}|} (s_j^{\prime} +1) \prod_{k=1}^{p+1}
      \left( f^{(k)} \right)^{s_k^{\prime}} \right].
  \end{align*}
  
  If $s$ is in $\mathscr{S}_{p}$, then $s+e_1 \in \mathscr{S}_{p+1}$
  and if moreover $s_j\neq 0$ then $s-e_j+e_{j+1}\in
  \mathscr{S}_{p+1}$. So the result is proved and the constants
  $\beta$ are defined by: $\beta_{\alpha,(1,0,0,\ldots)}=\alpha$; and
  if $s\in \mathcal{S}_p$ with $p>1$, then
  \begin{multline*}
    \beta_{\alpha,s}=\mathbb{1}_{s_1\neq0}(\alpha-|s|+1)\beta_{\alpha,(s_1-1,s_2,s_3,\ldots)}\\
    +\underset{s_j\neq0}{\sum_{2\leq j\leq
        p}}(s_{j-1}+1)\beta_{\alpha,(s_1,\ldots,s_{j-1},s_{j-1}+1,s_j-1,s_{j+1},\ldots)}.
  \end{multline*}
\end{proof}

\subsection{Several models}

In \cite{GKPS11}, the authors gave an $t^{O(m2^k})$ bound on the
number of distinct real roots of polynomials of the form $f=
\underset{i=1}{\overset{k}{\sum}}a_i\underset{j=1}{\overset{m}{\prod}}f_j^{\alpha_{i,j}}$,
where the $f_j$ are polynomials with at most $t$ monomials.  We
improve their result in Theorem~\ref{Thm_model1} using our results on
the Wronskian from Section~\ref{Sec_upperbound}.

\begin{lemma}\label{lem_monomials}
  Let $M$ be a set of $T$ monomials and $f_1,\ldots,f_s$ be
  polynomials whose monomials are in $M$.  For every formal monomial
  $P$ in the $s^2$ variables $f_1,f_1^\prime,\ldots
  f_1^{s-1},f_2,f_2^\prime,\ldots,f_s^{s-1}$ of degree $d$ and of
  order of differentiation $e$, the number of monomials in $x$ of
  $P(f_1,f_1^\prime,\ldots,f_s^{s-1})(x)$ is bounded by
  $\binom{d+T-1}{T-1}$. More precisely, the set of these monomials is
  included in a set $E_{d,e}$ of size at most $\binom{d+T-1}{T-1}$
  which does not depend on $P$.
\end{lemma}

\begin{proof}
  Let $M^d$ be the set of monomials which are the product of $d$ not
  necessarily distinct monomials of $M$. The cardinal of this set is
  bounded by the cardinal of the set of multisets of size $d$ of
  elements in $M$, that is $\binom{T+d-1}{T-1}$. It is easy to see
  that we can take the set $E_{d,e}$ defined as the set of monomials
  of $x^{-e}M^d$. Its cardinal is bounded by the cardinality of $M^d$.
\end{proof}

\begin{theorem} \label{Thm_model1} Let $f=
  \underset{i=1}{\overset{k}{\sum}}a_i\underset{j=1}{\overset{m}{\prod}}f_j^{\alpha_{i,j}}$
  be a non identically zero function such that each $f_j$ is
  a polynomial with at most $t$ monomials and such that
  $a_i\in\mathbb{R}$ and $\alpha_{i,j}\in\mathbb{N}$.  Then,
  $Z_{\mathbb{R}}(f) \leq
  4ktm+4\left(\e(1+t)\right)^{\frac{mk^2}{2}}=O(t^{\frac{mk^2}2})$.

  Moreover, if $I$ is a real interval such that for all $j$, $f_j(I)
  \subseteq ]0,+\infty[$ (which ensures $f$ is defined on $I$), then
  the result is still true for real (possibly negative) powers
  $\alpha_{i,j}$, i.e., $Z_I(f) \leq
  4ktm+4\left(\e(1+t)\right)^{\frac{mk^2}{2}}$.
\end{theorem}

\begin{proof}
  Let $N$ an integer such that for all $i$ and $j$, we have
  $\alpha_{i,j}+N >0$.  Let us consider $\tilde{f}= \sum_{i=1}^k a_i
  g_i$ where $g_i = \prod_{j=1}^m f_j^{\alpha_{i,j}+N+k}$. Note that
  $\tilde{f}=f \cdot \prod_{j=1}^m f_j^{N+k}$. We are going to bound
  the number of zeros of $\tilde{f}$.  In both cases (whether
  $\alpha_{i,j}$ are integer or real numbers), the functions $g_i$ are
  analytic in $I$. Furthermore, we can assume without loss of
  generality that the family $(g_i)$ is linearly independent. Indeed,
  if it is not the case, we can consider a basis of the family $(g_i)$
  and write $\tilde{f}$ in this basis. Then we can suppose that all
  $a_i$ are non-zero, otherwise, we remove these terms from the
  sum. We want to bound the number of zeros of $W(g_1, \ldots , g_s)$
  for all $s \leq k$ to conclude with Theorem~\ref{thm_main3}. We know
  that for $1 \le u,v \le s$
  \begin{align}
    g_u^{(v-1)} & = \sum_{\underset{r_1+ \ldots + r_m = v-1}{r_1,r_2,
        \ldots, r_m}} \prod_{j=1}^m \left( f_j^{\alpha_{u,j}+N+k}
    \right)^{(r_j)}.
  \end{align}
  We use now Lemma~\ref{lem_power} and we simplify the notation by
  writing $\beta_{u,j,s}$ instead of $\beta_{\alpha_{u,j}+N+k,s}$.
  \begin{align} \label{eqn_factor}
    g_u^{(v-1)} & = \sum_{\underset{r_1+ \ldots + r_m = v-1}{r_1,r_2, \ldots, r_m}} \prod_{j=1}^m \left[ \sum_{s \in \mathscr{S}_{r_j}} \beta_{u,j,s} f_j^{\alpha_{u,j}+N+k-|s|} \prod_{k=1}^{r_j} \left( f_j^{(k)} \right)^{s_k} \right] \\
    & = \left( \prod_{j=1}^m f_j^{\alpha_{u,j}+N} \right) \left(
      \prod_{j=1}^m f_j^{k-v+1} \right) T_{u,v}\left(
      (f_p^{(q-1)})_{1\leq p,q \leq s}\right). \nonumber
  \end{align}
  with : \[ T_{u,v}\left( (f_p^{(q-1)})_{1\leq p,q \leq s}\right) =
  \sum_{\underset{r_1+ \ldots + r_m = v-1}{r_1,r_2, \ldots, r_m}}
  \prod_{j=1}^m \left[ \sum_{s \in \mathscr{S}_{r_j}} \beta_{u,j,s}
    f_j^{v-1-|s|} \prod_{k=1}^{r_j} \left( f_j^{(k)} \right)^{s_k}
  \right]. \] The polynomial $T_{u,v}$ is homogeneous of total degree
  $(v-1)m$ with respect to the $s^2$ variables $\left( f_p^{(q-1)}
  \right)_{1 \leq p,q \leq s}$ and each of its terms is of
  differentiation order $v-1$.
  
  Then, we notice that, in~(\ref{eqn_factor}), the first parenthesis
  does not depend on $v$ and the second one on $u$. We get
  \begin{multline*}
    W(g_1, \ldots , g_s) = \\
    \left( \prod_{i=1}^s \prod_{j=1}^m f_j^{\alpha_{i,j}+N+k-i+1}
    \right) \det \left( \left(T_{u,v}\left( (f_p^{(q-1)})_{1\leq p,q
            \leq s}\right) \right)_{u,v \leq s}\right).
  \end{multline*}
  Hence,
  \begin{align} \label{Borne_zeros} & Z(W(g_1, \ldots , g_s)) \leq
    \left( \sum_{j=1}^m Z(f_j) \right) + Z\left(\det
      \left(T_{u,v}\left( (f_p^{(q-1)})_{1\leq p,q \leq s}\right)
      \right)\right).
  \end{align}

  We are now going to bound the number of monomials in $x$ of
  $\rm{det}(T_{u,v})$. We saw that $T_{u,v}$ is a homogeneous
  polynomial of degree $(v-1)m$ with respect to the $s^2$ variables
  $\left( f_p^{(q-1)} \right)_{1 \leq p,q \leq s}$ and of order of
  differentiation $v-1$. Moreover, as the family $(g_i)$ is linearly
  independent and as these functions are analytic, the Wronskian is
  not identically zero (Lemma~\ref{lem_zero}).  So $\rm{det}(T_{u,v})$
  is a linear combination, with respect to the variables $\left(
    f_p^{(q-1)} \right)_{1 \leq p,q \leq s}$, of monomials of degree
  exactly $\underset{v=1}{\overset{s}{\sum}}(v-1)m=m\binom{s}{2}$ and
  of order of differentiation $\binom{s}{2}$. By
  Lemma~\ref{lem_monomials}, the monomials in $x$ of each term of
  $\rm{det}(T_{u,v})$ are in the set
  $E_{\binom{s}{2}m,\binom{s}{2}}$. Consequently, the number of
  monomials in $x$ of $\rm{det}(T_{u,v})$ is bounded by the cardinal
  of $E_{\binom{s}{2}m,\binom{s}{2}}$, i$.$e$.$ by
  $\binom{m\binom{s}{2}+mt-1}{mt-1}$. Descartes' rule of signs
  (Lemma~\ref{lem_Descartes}) gives
\begin{align} \label{for_Upsilon}
  Z\left(\underset{u,v\leq s}{\rm{det}}\left(T_{u,v}\right)\right) \leq 2\binom{m\binom{s}{2}+mt-1}{mt-1}-1.
\end{align}

    We have now all the tools to prove the theorem. We have:
    \begin{align*}
      Z(f) & \leq Z(\sum_{i=1}^ka_ig_i).
\end{align*}
By Theorem~\ref{thm_main3}:
\begin{align*}      
  Z(f) & \leq k-1 + 2\sum_{s=1}^{k}Z(W(g_1,\ldots,g_s)).
\end{align*}
Using  formula~(\ref{Borne_zeros}):
\begin{align*}      
  Z(f) & \leq k-1+2k\left( \sum_{j=1}^m Z(f_j) \right) +
  2\sum_{s=1}^kZ\left(\underset{u,v \leq s}{\det} \left( T_{u,v}\left(
        (f_p^{(q-1)})_{p,q \leq s}\right) \right)\right).
\end{align*}
%Then using again 
By Descartes' rule, $\sum_{j=1}^m Z(f_j) \leq (2t-1)m$.  We can then
apply~(\ref{for_Upsilon}) to obtain the inequality
\begin{align*}
  Z(f) & \leq
  k-1+2k(2t-1)m+2\sum_{s=1}^k\left(2\binom{m\binom{s}{2}+mt-1}{mt-1}-1\right).
\end{align*}
Finally, we use the well known bound: $\binom{n}{k}\leq \left(\e n/k\right)^k$
\begin{align*}
  Z(f) & \leq k-1+4ktm-2km-2k+4+4\sum_{s=2}^k\left(\e\left(1+\frac{mt-1}{m\binom{s}{2}}\right)\right)^{m\binom{s}{2}}\\
  & \leq 4ktm+4\left(\e(1+t)\right)^{\frac{mk^2}{2}}.
  \end{align*}
\end{proof}

Using polynomials of small degrees instead of sparse polynomials, the
same argument gives a polynomial bound.

\begin{theorem} \label{Thm_model2} Let $f = \sum_{i=1}^k a_i
  \prod_{j=1}^m f_j^{\alpha_{i,j}}$ where $f$ is not null, the $f_j$
  are of degrees bounded by $d$ and such that the $a_i$ are reals and
  the $\alpha_{i,j}$ are integers.  Then, $Z_{\mathbb{R}}(f) \leq
  \frac{1}{3}k^3md +2kmd+k \sim \frac{k^3md}{3}$

  Moreover, if $I$ is a real interval such that for all $j$, $f_j(I)
  \subseteq \mathbb{R}^{+*\star}$ (which ensures $f$ is defined on
  $I$), then the result is always true for real powers $\alpha_{i,j}$,
  i$.$e$.$ $Z_I(f) \leq \frac{1}{3}k^3md +2kmd+k\sim \frac{k^3md}{3}$.
\end{theorem}

\begin{proof}
  In the proof of Theorem~\ref{Thm_model1}, we saw that $\det \left(
    T_{u,v}\right)$ is a homogeneous polynomial of degree $m
  \binom{s}{2}$ in the $s^2$ variables
  $f_1,f_1^\prime,\ldots,f_s^{(s-1)}$. So, it is of degree $md
  \binom{s}{2}$ in the variable $x$. Moreover, as the family $(g_i)$
  is linearly independent and as these functions are analytic, the
  Wronskian is not identically zero (Lemma~\ref{lem_zero}).  In
  Equation~(\ref{Borne_zeros}), the first term is bounded by $md$ and
  the second one by $md \binom{s}{2}$.  By
  Theorem~\ref{thm_main3}\footnote{Using Theorem~\ref{thm_main}
    instead of Theorem~\ref{thm_main3} would multiply our upper bound
    by a $O(k)$ factor.}, the number of zeros of $\sum_{i=1}^k a_ig_i$
  is bounded by $\left( \frac{1}{3}k^3md +2kmd+k\right)$.
\end{proof}

Avenda\~no studied the case $f= \sum_{i=1}^k
x^{\alpha_i}(ax+b)^{\beta_i}$ where $\alpha_i$ and $\beta_i$ are
integers \cite{Ave09}. He found an upper bound linear in $k$ for the
number of roots. But he showed also that his bound is false in the
case of real powers. We find here a polynomial bound which works also
for real powers.

\begin{corollary}\label{Cor_Avendano}
  Let $f= \sum_{i=1}^k c_ix^{\alpha_i}(ax+b)^{\beta_i}$. Let $I$ be
  the interval $\{x \in \mathbb{R} | x>0 \wedge ax+b > 0\}$.  Then
  $Z_I(f)\leq \frac{2}{3}k^3+5k$.
\end{corollary}

Li, Rojas and Wang \cite{LRW03}(Lemma 2) showed that polynomials $f =
\sum_{i=1}^k a_i \prod_{j=1}^m (c_jX+d_j)^{\alpha_{i,j}}$ where
coefficients $a_i$, $b_j$, $c_j$ and exponents $\alpha_{i,j}$ are
real, have at more $m+m^2+\ldots+m^{k-1}$ zeros. Our result improves
this bound:

\begin{corollary}
  Let $f= \sum_{i=1}^k a_i \prod_{j=1}^m (c_jx+d_j)^{\alpha_{i,j}}$
  where the coefficients $a_i$,$c_j$, $d_j$ and the exponents
  $\alpha_{i,j}$ are real numbers. On the interval $I=\{x \in
  \mathbb{R} | \forall j, c_jx+d_j > 0\}$
%which is an interval. 
we have $Z_I(f) = O(mk^3)$.
\end{corollary}

Another corollary was suggested to us by Maurice
Rojas. In~\cite{LRW03}, Li, Rojas and Wang bound, when it is finite,
the number of intersection between a trinomial curve and a $t$-sparse
curve by $2^t-2$. We improve here their result. The main idea is to
make a change of variables and reduce to the case where $f$ is
affine. This may introduce rational exponents, even if the original
system has integer coefficients only.

\begin{corollary}\label{Cor_trinomials}
  Let $f$ be a non-zero bivariate trinomial and $g$ be a bivariate $t$-sparse
  polynomial. Then the number of positive intersections between these two
  curves is infinite or bounded by $\frac{2}{3}t^3+5t$.
  
  Furthermore, the result still holds if the coefficients are real.
\end{corollary}

\begin{proof}
  Let
  $f(X,Y)=c_1X^{\gamma_1}Y^{\delta_1}+c_2X^{\gamma_2}Y^{\delta_2}+c_3X^{\gamma_3}Y^{\delta_3}$
  (with $c_3\neq 0$) and
  $g(X,Y)=\sum_{i=1}^ta_iX^{\alpha_i}Y^{\beta_i}$. On
  $\left(\mathbb{R}^{+\star}\right)^2$, the zeros of $f$ are the same
  than the zeros of
  $c_1+c_2X^{\gamma_2-\gamma_1}Y^{\delta_2-\delta_1}+c_3X^{\gamma_3-\gamma_1}Y^{\delta_3-\delta_1}$. Then
  we can and will assume that $\gamma_1=\delta_1=0$.
  \begin{itemize}
  \item First case: there exists $r\in\mathbb{R}^\star$ such that
    $(\gamma_3,\delta_3)=r\cdot
    (\gamma_2,\delta_2)$. In this case, we can put
    $A=X^{\gamma_2}Y^{\delta_2}$, then on $\mathbb{R}^{+\star}$:
    \begin{align*}
      f=0 & \Leftrightarrow c_1+c_2A+c_3A^r =0.
    \end{align*}
    By Descartes' rule of signs (Lemma~\ref{lem_Descartes}): the last
    equation has at most two real positive solutions which will be
    denoted $s_1$ and $s_2$. So the system is equivalent to the
    following:
    \begin{align*}
      \begin{cases}
        Y=\left(\frac{s_1}{X^{\gamma_2}}\right)^{\frac{1}{\delta_2}}
        \textrm{ or
        }Y=\left(\frac{s_2}{X^{\gamma_2}}\right)^{\frac{1}{\delta_2}}
        \\
        g(X,Y)=0.
      \end{cases}
    \end{align*}
    We obtain
    \begin{align*}
      g(X,Y)=\sum_{i=1}^ta_i s_1^{\frac{\beta_i}{\delta_2}}
      X^{\alpha_i-\frac{\beta_i\gamma_2}{\delta_2}} \textrm{ or } \sum_{i=1}^ta_i s_2^{\frac{\beta_i}{\delta_2}}
      X^{\alpha_i-\frac{\beta_i\gamma_2}{\delta_2}} .
    \end{align*}
    Again by Lemma~\ref{lem_Descartes}, the number of positive
    solutions is infinite or bounded by $2t-2\leq \frac{2}{3}t^3+5t$.
  \item Second case: the family $\left((\gamma_2,\delta_2),
      (\gamma_3,\delta_3)\right)$ is linearly independent. We can
    define $A=X^{\gamma_2}Y^{\delta_2}$ and
    $B=X^{\gamma_3}Y^{\delta_3}$, then on $\mathbb{R}^{+\star}$:
    \begin{align*}
      f=0 & \Leftrightarrow c_1+c_2A+c_3B =0\\
      & \Leftrightarrow   B =-\frac{c_1}{c_3}-\frac{c_2}{c_3}A.
    \end{align*}
    Let us define $\Delta=\det
    \begin{vmatrix}
      \gamma_2 & \gamma_3 \\
      \delta_2 & \delta_3
    \end{vmatrix} \neq 0
    $. Then
    $X=A^{\frac{\delta_3}{\Delta}}B^{-\frac{\delta_2}{\Delta}}$
    and
    $Y=A^{-\frac{\gamma_3}{\Delta}}B^{\frac{\gamma_2}{\Delta}}$. Hence,
    \begin{align*}
      g(X,Y)=0 \Leftrightarrow \sum_{i=1}^t a_i
      A^{\frac{\alpha_i\delta_3-\beta_i\gamma_3}{\Delta}}B^{\frac{-\alpha_i\delta_2+\beta_i\gamma_2}{\Delta}}=0.
    \end{align*} The number of solutions corresponds to the number of
    roots of the polynomial: 
    \begin{align*}
      \sum_{i=1}^t a_i
      A^{\frac{\alpha_i\delta_3-\beta_i\gamma_3}{\Delta}}\left(-\frac{c_1}{c_3}-\frac{c_2}{c_3}A\right)^{\frac{-\alpha_i\delta_2+\beta_i\gamma_2}{\Delta}}=0.
    \end{align*} By Corollary~\ref{Cor_Avendano}, the number of positive
    roots is infinite (if the polynomial is identically zero) or
    bounded by $\frac{2}{3}t^3+5t$.
  \end{itemize}
\end{proof}

\section{Some Algorithms for Polynomial Identity
  Testing}\label{Sec_PIT}

A PIT algorithm takes a polynomial as input, and decides whether the
polynomial is identically equal to zero. There are two classical forms
for these algorithms: blackbox and whitebox. For the first one, the
input is given by a blackbox. And in the second case, the input is
given by a circuit. These two types of algorithms are not comparable
in our case since, if we have a circuit, we cannot always evaluate it
efficiently on an input because the circuit may be of high degree.

\subsection{Blackbox PIT algorithms}

The bounds on real roots of Theorem~\ref{Thm_model1} immediately give
a blackbox PIT algorithm for some families of polynomials.

\begin{corollary}\label{cor_blackbox1}
  Let $f=
  \underset{i=1}{\overset{k}{\sum}}a_i\underset{j=1}{\overset{m}{\prod}}f_j^{\alpha_{i,j}}$
  be a function such that $f_j$ is a polynomial with at most $t$
  monomials and such that $a_i\in\mathbb{R}$ and
  $\alpha_{i,j}\in\mathbb{N}$.  Then, there is a blackbox PIT
  algorithm which makes only
  $1+4ktm+4\left(e(1+t)\right)^{\frac{mk^2}2}$ queries.
\end{corollary}

\begin{proof}
  We consider the algorithm which tests if the polynomial outputs zero
  on the $1+4ktm+4\left(e(1+t)\right)^{\frac{mk^2}2}$ first
  integers. By Theorem~\ref{Thm_model1}, this set is a hitting set,
  that is, if the polynomial is not zero, then at least one of these
  integers will not be a root of the polynomial.
\end{proof}

\begin{corollary}\label{cor_blackbox2}
  Let $f = \sum_{i=1}^k a_i \prod_{j=1}^m f_j^{\alpha_{i,j}}$ where
  each $f_j$ is of degree bounded by $d$ and such that the $a_i$ are reals
  and the $\alpha_{i,j}$ are integers.  Then, there is a blackbox PIT
  algorithm which makes only $1+\frac13k^3md+2kmd+k$ queries.
\end{corollary}

\begin{proof}
  We apply Theorem~\ref{Thm_model2}.
\end{proof}

\subsection{A Whitebox PIT algorithm}

Results for whiteboxes are more complicated. They rely on the link
between Wronskians and linear independence. In this section, we prove
the following proposition:

\begin{proposition} \label{pro_whitebox1} Let $f=
  \underset{i=1}{\overset{k}{\sum}}a_i\underset{j=1}{\overset{m}{\prod}}f_j^{\alpha_{i,j}}$
  where $f_j$ is a polynomial with at most $t$ monomials, the $a_i$
  are integers and the $\alpha_{i,j}$ are non-negative integers. Let
  $C$ be an upper bound on the degrees of the $f_j$, on the bit size
  of their coefficients as well as on the bit size of the coefficients
  $a_i$ and of the exponents $\alpha_{i,j}$.  Then, there is a
  whitebox PIT algorithm which decides if $f$ is zero in time
  $\tilde{O}\left(C2^{4mk^2\log t}\right)$.
\end{proposition}

First, we will need an algorithm for testing if some Wronskians are
identically zero or not.  We now describe an algorithm which takes as
inputs functions
$h_1=\prod_{j=1}^m(f_j)^{\alpha_{1,j}},\ldots,h_l=\prod_{j=1}^m(f_j)^{\alpha_{l,j}}$
(given by sequences $(f_j)_{1\leq j\leq m}$ and $(\alpha_{i,j})_{1\leq
  i\leq l, 1\leq j\leq m}$) and which outputs the leading coefficient
of the Wronskian $W(h_1,\ldots,h_l)$ if this determinant is not
identically zero, and outputs zero otherwise.  We will use the
notation $f(n)=\tilde{O}(g(n))$. It is a shorthand for $f(n) = O(g(n)
\log^k g(n))$ for some constant $k$.

\begin{proposition}\label{pro_algo} 
  There is an algorithm which on the input $(f_j)_{j\leq m}$,
  $(\alpha_{i,j})_{i\leq l,j\leq m}$ outputs the leading coefficient
  of the Wronskian of
  $$\left(\prod_{j=1}^m(f_j)^{\alpha_{1,j}},\ldots,\prod_{j=1}^m(f_j)^{\alpha_{l,j}}\right)$$
if the Wronskian is not identically zero and which outputs zero
otherwise. This algorithm runs in time $\tilde{O}\left(C2^{4ml^2\log
    t}\right)$.
\end{proposition}

\begin{proof}
  As in the proof of Theorem~\ref{Thm_model1}, we in fact compute the
  Wronskian of $\left(g_1,\ldots,g_l\right)$ where
  $g_i=\prod_{j=1}^m\left(f_j\right)^{\alpha_{i,j}+l}$. To get the
  correct leading coefficient, we can just notice that:
\begin{align*}
  W\left(g_1,\ldots,g_l\right)=W(h_1,\ldots,h_l)\prod_{j=1}^m\left(f_j\right)^{l^2}.
\end{align*}

Hence, we want to compute the Wronskian of
$\left(g_1,\ldots,g_l\right) $. Again as in the proof of
Theorem~\ref{Thm_model1}, we factorize each column $u$ by
$\prod_{j=1}^m\left(f_j\right)^{\alpha_{u,j}}$ and each row $v$ by
$\prod_{j=1}^m\left(f_j\right)^{l-v+1}$. We will denote the resulting
matrix by $M$. The entries of this matrix are polynomials.

According to Lemma~\ref{lem_cell} in Appendix~\ref{Sec_proof_PIT}, we
can compute the expanded polynomial in one cell $(v,u)$ of $M$ in time
$\tilde{O}\left(2^{vm}t^{mv}v^mC\log l\right)$.

Then, computing all entries of the matrix which is of size $(l\times
l)$ needs $\tilde{O}\left(2^{lm}t^{ml}l^mC\right)$ operations. Next,
we have to compute the determinant of this matrix. We are going to
compute this determinant directly by expanding it as a sum of $l!$
products. This computation takes time $\tilde{O}\left(C2^{4ml^2\log
    t}\right)$ by Lemma~\ref{lem_determinant} in
Appendix~\ref{Sec_proof_PIT}.

If the determinant is zero, it means that the Wronskian is zero, then
the algorithm outputs zero. Otherwise, for computing the leading
coefficient, we have to multiply the coefficient we got by the leading
coefficient of
$\frac{\left(\prod_{u=1}^l\prod_{j=1}^m(f_j)^{\alpha_{u,j}}\right)\left(\prod_{v=1}^l\prod_{j=1}^m(f_j)^{l-v+1}\right)}{\prod_{j=1}^m(f_j)^{l^2}}$. This
operation takes $\tilde{O}\left(Cml(C+l)\right)$ operations since we
can compute the product of $n$ integers of size $s$ in time
$\tilde{O}(ns)$.
\end{proof}

Second, we will also need the following algorithm: if
$W(h_1,\ldots,h_l)\neq 0$ and $W(h_1,\ldots,h_{l+1})=0$ then find
$a_1,\ldots,a_{l+1}$ such that $a_1h_1+\ldots+a_lh_l=h_{l+1}$ (these
constants exist according to Lemma~\ref{lem_zero}). So for each
$i\in\mathbb{N}$, $a_1h_1^{(i)}+\ldots+a_lh_l^{(i)}=h_{l+1}^{(i)}$. We
can compute the $a_j$ using Cramer's formula.  As a result, for each
$1\leq j\leq l$ we have:
 
\begin{align*}
a_j = & \frac{
  \begin{vmatrix}  
    h_1 & \cdots & h_{j-1} & h_{l+1} & h_{j+1} & \cdots & h_l\\
    h_1^{\prime} & \cdots & h_{j-1}^{\prime} & h_{l+1}^{\prime} & h_{j+1}^{\prime} & \cdots & h_l^{\prime}\\
    \vdots &\ddots&\vdots&\vdots&\vdots&\ddots&\vdots\\
    h_1^{(l-1)} & \cdots & h_{j-1}^{(l-1)} & h_{l+1}^{(l-1)} &
    h_{j+1}^{(l-1)} & \cdots & h_l^{(l-1)}
  \end{vmatrix}
}{W(h_1,\ldots,h_l)} \\
=&\frac{lc\left(W(h_1,\ldots,h_{j-1},h_{l+1},h_{j+1},\ldots,h_l)\right)}{lc\left(W(h_1,\ldots,h_l)\right)}
\end{align*}
where $lc(W(h_1,\ldots,h_l))$ is the leading coefficient of the
Wronskian for the family $(h_1,\ldots,h_l)$.  The previous algorithm
(Proposition~\ref{pro_algo}) computes these coefficients, so we can
compute the $(a_j)$ in time $\tilde{O}\left(C2^{4ml^2\log t}\right)$.

Finally,  with such algorithms, we just need to go from $a_1h_1$ to
$a_1h_1+\ldots+a_kh_k$. Each time we add a $h_i$, either it is linearly
independent and we add it to the current basis or it is dependent, and we
write it in the current basis. At the end, $a_1h_1+\ldots+a_kh_k$
is expressed as a linear combination of basis functions. We just have to check if
all coefficients are zero to conclude if this function is identically zero or
not. This completes the proof of Proposition~\ref{pro_whitebox1}.

\section{An improved upper bound}\label{Sec_upperbound_refinement}

In this section, we give a proof of Theorem~\ref{thm_main3}.  First,
we point out that Voorhoeve and van der Poorten's paper~\cite{VP75}
contains a result similar to Theorem~\ref{thm_main}, except that all
zeros are counted with multiplicity.

\begin{theorem}\label{Thm_Voor}[Voorhoeve and van der Poorten, 1975]
  Let $f_1,\ldots,f_k$ be real analytic functions over an interval
  $I$. Then,
  \begin{align*}
    N(f_1+\ldots + f_k) \leq k-1
    +\sum_{j=1}^{k-2}N(W_j)+\sum_{j=1}^{k}N(W_j)
  \end{align*}
  where $N(f)$ is the number of roots of $f$ on $I$ counted with
  multiplicities.
\end{theorem}

This result immediately implies Theorem~\ref{Thm_heart} in the
analytic case. However in our applications we will have to not
consider the multiplicity of roots. Indeed, the bounds in
Theorem~\ref{Thm_model1} and~\ref{Thm_model2} do not depend on the
exponents $\alpha_{i,j}$. If we counted multiplicities, the resulting
bound would depend on the $\alpha_{i,j}$. Using some ideas of the
proof of this theorem, we can nevertheless improve
equation~(\ref{Eq1}).

We will denote $W_i=W(f_1,f_2,\ldots,f_i)$ for $i\geq 1$ when the
family $(f_1,\ldots,f_i)$ is clear from the context.  Finally, we
define $W_0=1$.

In addition to Lemma~\ref{lem_zero}, several connections between the
Wronskian and the linear combination of the functions are known.  We
will use a result of Frobenius \cite{Fro76, Pol22,NY95}:

\begin{lemma}\label{lem_Frobenius}
  Let $f_i$ a family of analytic functions.  Let $R_i$ be the family
  of functions defined by:
  \[R_0 = f_1+\ldots+f_k\]
  \[R_{i+1} = \frac{W_{i+1}^2}{W_i}\left(\frac{R_i}{W_{i+1}}\right)^\prime.\]

Then the functions $R_i$ are analytic and $R_{k-1} = W_k$.
\end{lemma}

\begin{proof}[Proof of Theorem~\ref{thm_main3}]
  Let $R_i$ be the family of analytic functions defined by:
  \[R_0 = f_1+\ldots+f_k\]
  \[R_{i+1} = \frac{W_{i+1}^2}{W_i}\left(\frac{R_i}{W_{i+1}}\right)^\prime.\]
  
  We will prove by induction that for all $1\leq i\leq k-1$, the
  analytic function $R_i$ has at least
  $Z(f_1+\ldots+f_k)-i-Z(W_i)-2\sum_{j=1}^{i-1}Z(W_j)$ roots on
  $I$. That yields the theorem with $i=k-1$ and
  Lemma~\ref{lem_Frobenius}.
  
  If $i=0$, then $R_0 = f_1+\ldots+f_k$ and $R_0$ has exactly (and so
  at least) $Z(f_1+\ldots+f_k)$ zeros.
  
  Suppose that the property is verified for a particular $i$ such that
  $i\leq k-2$.  We will denote $m_x(F)$ the multiplicity of the root
  $x$ in $F$ for all $x\in \mathbb{R}$ and analytic function
  $F$. We define
  four values:
  \begin{itemize}
  \item $Z_i^+$ is the number of $x\in \mathbb{R}$ such that
    $m_x(R_i) > m_x(W_{i+1})>0$.
  \item $Z_i^=$ is the number of $x\in \mathbb{R}$ such that
    $m_x(R_i) = m_x(W_{i+1})>0$.
  \item $Z_i^{->}$ is the number of $x\in \mathbb{R}$ such that
    $m_x(W_{i+1})>m_x(R_i) >0$.
  \item $Z_i^{-0}$ is the number of $x\in \mathbb{R}$ such that
    $m_x(W_{i+1})>0 = m_x(R_i) $.
  \end{itemize}
  
  We have: $Z(W_{i+1})=Z_i^+ +Z_i^=+Z_i^{-0}+Z_i^{->}$. 

  We know by Lemma~\ref{lem_Frobenius} that $R_i$ is indeed analytic.
  Then by induction hypothesis, the fraction $\frac{R_i}{W_{i+1}}$ has
  at least
  \begin{align*}
    Z(f_1+\ldots+f_k) -i-Z(W_i) -2\left(\sum_{j=1}^{i-1}Z(W_j)\right)
    -Z_i^=-Z_i^{->}
  \end{align*} 
  roots and at most $Z_i^{->}+Z_i^{-0}$ poles. By Rolle's Theorem, the
  number of zeros of $\left(\frac{R_i}{W_{i+1}}\right)^\prime$ is at
  least
  \begin{multline*}
    \left[Z(f_1+\ldots+f_k)-i-Z(W_i)-2\left(\sum_{j=1}^{i-1}Z(W_j)\right)
      -Z_i^=-Z_i^{->}\right] \\ 
    -\left[Z_i^{->}+Z_i^{-0}+1\right].
  \end{multline*} 
  So, the number of zeros of $R_{i+1} =
  \frac{W_{i+1}^2}{W_i}\left(\frac{R_i}{W_{i+1}}\right)^\prime$ is at
  least
  \begin{align*}
    & \left[Z(f_1+\ldots+f_k)-i-Z(W_i)-2\left(\sum_{j=1}^{i+1}Z(W_j)\right)-Z_i^=-Z_i^{->}\right]\\
    & -\left[Z_i^{->}+Z_i^{-0}+1\right]+Z_i^{->}-Z(W_i).
  \end{align*}
  We used here that if $x$ is such that $0<m_x(R_i)<m_x(W_{i+1})$ then
  $-m_x(W_{i+1})<m_x\left(R_i)-m_x(W_{i+1})\right)<0$ and so
    $m_x\left(W_{i+1}^2\left(\frac{R_i}{W_{i+1}}\right)^\prime\right)\geq
    m_x(W_{i+1})-1>0$. Hence
  \begin{align*}
    Z(R_{i+1})\geq Z(f_1+\ldots+f_k)-(i+1)
    -Z(W_{i+1})-2\left(\sum_{j=1}^iZ(W_j)\right).
  \end{align*}
\end{proof}

\section{Optimality of Theorem~\ref{thm_main}} \label{Sec_optimality}

Recall that in Theorem~\ref{thm_main}, it was proved
that $$Z(a_1f_1+\ldots+a_kf_k)\leq(1+|\Upsilon|)k$$ with $\Upsilon =
\underset{1\leq i\leq k}{\bigcup}Z\left(W(f_1,\ldots,f_i)\right)$. It
will be shown in Theorem~\ref{thm_opt} that this theorem is quite
optimal in the sense that for arbitrarily large values of $|\Upsilon|$
and $k$, we can find functions $f_1,\ldots,f_k$ and coefficients
$a_1,\ldots,a_k$ such
that $$Z(a_1f_1+\ldots+a_kf_k)\geq(1+|\Upsilon|)(k-1).$$

We begin by proving a technical lemma.

\begin{lemma} \label{lem_derivopt} Let $f$ be a non-constant
  polynomial. There exists, for $0\leq i\leq q$, rational functions
  $F_{i,q}$ such that if we define the function
  $h_{p,i}=\frac{p!}{(p-i)!}\left(f^\prime\right)^if^{p-i}$, then we
  have:
\begin{enumerate}
\item for all $q \geq 0$, we have $F_{q,q}=1$
\item for all $0\leq i\leq q$, the rational function
  $\left[\left(f^\prime\right)^qF_{i,q}\right]$ is a polynomial
\item For all $q \geq 0$, for all $p \geq 1$ we have:
  $\left(f^p\right)^{(q)}=\underset{i=0}{\overset{q}{\sum}}h_{p,i}
  F_{i,q}$.
\end{enumerate}
\end{lemma}
The main point is that $F_{i,q}$ does not depend on $p$.
\begin{proof}
  We define $F_{i,q}$ by induction on $q$. If $q=0$, let us define
  $F_{0,0}=1$. Then, we have $f^p=h_{p,0}$ and
  $\left[\left(f^\prime\right)^0F_{0,0}\right] = 1$.

  We suppose now that $F_{i,q^\prime}$ are defined for all $i$,
  $q^\prime$ such that $0\leq i\leq q^\prime\leq q$. Let us define
  $F_{i,q+1}$.  We have:
  \begin{align*} 
    \left(h_{p,i}\right)^\prime & = \left[\frac{p!}{(p-i)!}(f^\prime)^if^{p-i}\right]^\prime \\
    & = \frac{p!}{(p-i)!} \left[(p-i) (f^\prime)^{i+1}f^{p-i-1} + i f^{\prime\prime}(f^\prime)^{i-1}f^{p-i} \right] \\
    & = h_{p,i+1} +
    \left(i\frac{f^{\prime\prime}}{f^\prime}\right)h_{p,i}.
  \end{align*}
  So,
  \begin{align*}
    \left(f^p\right)^{(q+1)} & =\left(\underset{i=0}{\overset{q}{\sum}}h_{p,i}F_{i,q} \right)^\prime \\
    & = \underset{i=0}{\overset{q}{\sum}}\left[
      h_{p,i}(F_{i,q})^\prime + \left(h_{p,i+1} +
        \left(i\frac{f^{\prime\prime}}{f^\prime}\right)h_{p,i}\right)
      F_{i,q}\right] \\
    & = h_{p,0} (F_{0,q})^\prime +
    \left[\underset{i=1}{\overset{q}{\sum}}
      h_{p,i}\left((F_{i,q})^\prime +
        F_{i,q}\left(i\frac{f^{\prime\prime}}{f^\prime}\right) +
        F_{i-1,q}\right) \right] \\
    & \qquad \qquad \qquad \qquad \qquad \qquad \qquad \qquad \qquad
    \qquad+h_{p,q+1} F_{q,q}.
  \end{align*}

We can then define 
  \begin{align*}
    & F_{0,q+1}= (F_{0,q})^\prime \\
    & \rm{for} \ 1\leq i\leq q,\  F_{i,q+1}=(F_{i,q})^\prime + 
    F_{i,q}\left(i\frac{f^{\prime\prime}}{f^\prime}\right)+F_{i-1,q} \\
    & \textrm{and } F_{q+1,q+1}=F_{q,q}=1 \textrm{.}
  \end{align*}

  Then, we have (1) and (3) by construction. Finally, (2) is verified,
  since by induction hypothesis:
  \begin{align*}
    (f^\prime)^{q+1}F_{i,q}^\prime=\left[f^\prime \left( (f^\prime )^q
        F_{i,q}\right)\right]^\prime - (q+1)f^\prime \left(
      (f^\prime)^q F_{i,q}\right)\textrm{ is a polynomial.}
  \end{align*}
\end{proof}

We are going to show that the zeros of $W\left( f^{\alpha_1+k},
  \ldots, f^{\alpha_s+k}(x)\right)$ are either zeros of $f$ or zeros
of $f^\prime$.

\begin{lemma} \label{lem_Zff} Let $f$ be an analytic non-constant
  function in an interval $I$ and $\alpha_0,\ldots,\alpha_k$ be $k$
  pairwise distinct integers (with $k\geq 1$).  Then
  \begin{align*}
    \left\{x \in I | \exists s \leq k, \: W\left( f^{\alpha_1+k},
        \ldots, f^{\alpha_s+k} \right)(x) = 0\right\} \subseteq
    \left\{x \in I | (ff^\prime)(x) = 0\right\}.
  \end{align*}
\end{lemma}

\begin{proof} 
  Let us consider $f^{\alpha_1+k}, \ldots, f^{\alpha_k+k}$. First
  suppose that this family is linearly dependent. This means that
  there exist some constants $(a_1,\ldots,a_k)\in
  \mathbb{R}^k\setminus \{(0,0,\ldots,0)\}$ such
  that \begin{equation}\label{eqn_lineardependence} \sum_{i=1}^ka_i
    f^{\alpha_i+k}=0 \textrm{ on }I.\end{equation} But the integers
  $(\alpha_i+k)$ are all distinct so the polynomial
  $P(Y)=\sum_{i=1}^ka_i Y^{\alpha_i+k}$ is not zero. Hence $P(Y)$ has
  a finite number of roots. By (\ref{eqn_lineardependence}),
  $\rm{Im}(f)$ is included in the (finite) set of roots of
  $P$. Nevertheless, as $f$ is continuous, by the intermediate value
  theorem $\rm{Im}(f)$ is a real interval. So $\rm{Im}(f)$ is a
  singleton. This contradicts the hypothesis that $f$ is not constant.
  Therefore, this family is linearly independent.

 Let $\Delta$ be the matrix defined by $\Delta_{i,j} = \left( f^{\alpha_i+k}\right)^{(j-1)}$.
  By Lemma~\ref{lem_derivopt}, we get 
  $\Delta_{i,j}= \underset{l=0}{\overset{j-1}{\sum}}  h_{\alpha_i+k,l} \ F_{l,j-1}$, i$.$e$.$ in terms of matrix product: 
  $$\Delta = \left[ h_{\alpha_i+k,l-1} \right]_{1\leq i,l \leq s}\left[ F_{l-1,j-1} \mathbb{1}_{l \leq j}\right]_{1\leq l,j \leq s}.$$
  The second matrix of the product is an upper triangular matrix whose entries on the main diagonal are $1$ and so its determinant is $1$.
  Then, 
  $$\det \left(\left(\Delta_{i,j}\right)_{1\leq i,j \leq s}\right) = \det \left(\left(h_{\alpha_i+k,j-1}\right)_{1\leq i,j \leq s}\right).$$
  Finally, $h_{\alpha_i+k,j-1}=\frac{(\alpha_i+k)!}{(\alpha_i+k-j+1)!}\left[f^{\alpha_i}\right]\left[\left(f^\prime\right)^{j-1}f^{k-j+1}\right]$. The first bracket does not depend on $j$ and the second one on $i$. 
  Consequently, 
  $$\det \left(h_{\alpha_i+k,j-1}\right) = \left[f^{\underset{l=1}{\overset{s}{\sum}}\alpha_l}\right]\left[\left(f^\prime\right)^{\binom{s}{2}}f^{(k-s+1)s+\binom{s}{2}}\right] \det \left(\frac{(\alpha_i+k)!}{(\alpha_i+k-j+1)!}\right).$$
  Then, for all $x$ in I :
  \begin{align}\label{eqn_impl}
    W\left( f^{\alpha_1+k}, \ldots, f^{\alpha_s+k} \right)(x) = 0  \Leftrightarrow \det \left(h_{\alpha_i+k,j-1}\right)(x) =0 \\
    \Rightarrow \left\{f(x)=0 \textrm{ or } f^\prime(x)=0 \textrm{ or
      }\det \left(\frac{(\alpha_i+k)!}{(\alpha_i+k-j+1)!}\right) = 0
    \right\}. \nonumber
\end{align}
If $\det \left(\frac{(\alpha_i+k)!}{(\alpha_i+k-j+1)!}\right) =0$, as
it does not depend on $x$, the function $\det
\left(h_{\alpha_i+k,j-1}\right)$ vanishes for all $x$ and so the
Wronskian is zero over $I$. But as the functions $f^{\alpha_i+k}$ are
analytic, they would be linearly dependent by
Lemma~\ref{lem_zero}. That contradicts the hypothesis.  Consequently,
  \begin{align*} 
    \left\{x \in I | \exists s \leq k, \: W\left(
        f^{\alpha_1+k}, \ldots, f^{\alpha_s+k} \right) = 0\right\}
    \subseteq \left\{x \in I | (ff^\prime)(x) = 0\right\}. 
  \end{align*}
\end{proof}

As a byproduct, we give another proof of the weak version of
Descartes' rule of signs (Lemma~\ref{lem_Descartes}).  Let
$g=\sum_{i=1}^ka_ix_i^{\alpha_i}$. We need to show that the number of
distinct real roots is bounded by $2k-1$.  We can use the result of
Lemma~\ref{lem_Zff} with $f(x)=x$. In this case $g = \sum_{i=1}^k
a_ix^{\alpha_{i}}$. We get
\begin{align*}
  \Upsilon = \left\{x \in I | \exists s \leq k, \: W\left(
      f^{\alpha_1+k}, \ldots, f^{\alpha_s+k} \right) = \{0\}\right\}
 & \subseteq \left\{x \in I | (ff^\prime)(x) = 0\right\} \\
 & \subseteq \{0\}
\end{align*}
So, Theorem~\ref{thm_main} gives: $Z(f) \leq 2k-1$. A similar proof of
Lemma~\ref{lem_strongDescartes} appears in \cite{PoSz76} (Part V,
exercise 90).

In Lemma~\ref{lem_Zff}, it can be seen that the converse of the
implication~(\ref{eqn_impl}) is true as soon as $f^\prime$ really
appears as a factor of $\det
\left(\left(h_{\alpha_i+k,j-1}\right)_{1\leq i,j \leq s}\right)$. It
is the case when $\binom{s}{2}$ is different from zero, that is to say
when $s\geq 2$.
%It proves that
This implies the following result.
\begin{lemma} \label{lem_easy} Let $f$ be an analytic non-constant
  function in an interval $I$ and $\alpha_0,\ldots,\alpha_k$ be $k$
  pairwise distinct integers with the condition $k\geq 2$.  Then
$$\left\{x \in I | \exists s \leq k, \: W\left( f^{\alpha_1+k}, \ldots, f^{\alpha_s+k} \right)(x) = 0\right\} = \left\{x \in I | (ff^\prime)(x) = 0\right\}.$$
\end{lemma}

We have now all the tools to prove the main result of the section: the
optimality of Theorem~\ref{thm_main}.

\begin{theorem}\label{thm_opt}
  Let $\Upsilon = \left\{x\in I | \exists i \leq k, W(f_1, \ldots,
    f_i)(x)=0 \right\}$ as in Theorem \ref{thm_main}.  For every $k$
  and $p$, there exists a function $g = \sum_{i=1}^k
  a_if^{\alpha_{i}}$ such that $\alpha_i$ are positive integers, $f$
  is a polynomial such that $|\Upsilon| \geq p$ and such that $g$ has
  at least $\left(1+|\Upsilon|\right)(k-1)+Z(f)$ zeros.
\end{theorem}

\begin{proof}

\begin{figure}
\centering
\setlength\fboxsep{0pt}
\setlength\fboxrule{0.5pt}
\fbox{\includegraphics[scale=0.65]{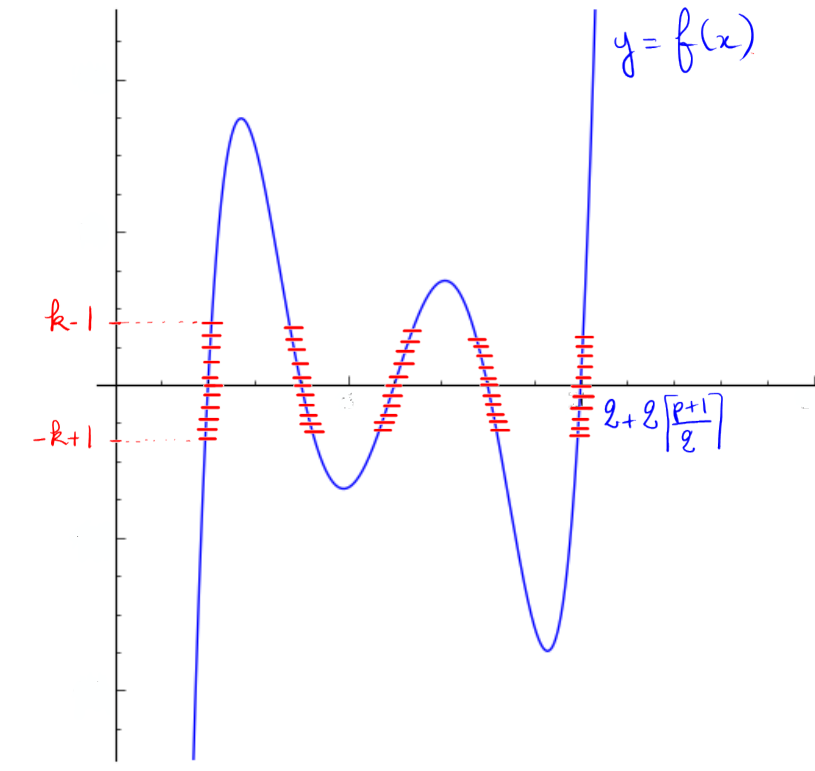}}
\caption{roots of $g=h\circ f$ in the proof of Theorem~\ref{thm_opt}}
\label{fig:nb-roots}
\end{figure}

  Let $h=x\prod_{i=1}^{k-1}\left(x^2-i^2\right)$. This polynomial is $k$-sparse and has $2k-1$ distinct real roots: $-k+1 < \ldots < -1 < 0 < 1 < k-1$.

  Let $f = k~\underset{i=1}{\overset{1+\left\lceil \frac{p+1}{2}  \right\rceil}{\prod}}(x-2i)$.

  Then, we just have to verify that $g = h \circ f$ has the required properties.

We have $g(x) =0$ if and only if $f(x)\in [-k+1,k-1]\cap\mathbb{Z}$. But for $y$ an odd integer, we have $|f(y)|> k-1$ and for $y$ an even integer between $2$ and $2+2\left\lceil \frac{p+1}{2}  \right\rceil $, we have $f(y)=0$. By the Intermediate Value Theorem, $g$ has at least $k-1$ zeros over each interval $(n,n+1)$ with $1 \leq n \leq 2+2\left\lceil \frac{p+1}{2}  \right\rceil$. So 

\begin{align} \label{eqn_z(g)}
  Z(g) & = 2\left(1+\left\lceil \frac{p+1}{2}  \right\rceil\right)
(k-1)+\left(1+\left\lceil \frac{p+1}{2}  \right\rceil\right) \nonumber
\\
& = (2k-1)\left(1+\left\lceil \frac{p+1}{2}  \right\rceil\right).
\end{align}

Rolle's Theorem ensures that for two roots of $f$, there exists a root of $f^\prime$ which is strictly between both roots of $f$. Hence, $Z(ff^\prime) \geq 2Z(f)-1 =1+2 \left\lceil \frac{p+1}{2}  \right\rceil$. Considering the degree of $ff^\prime$, we find $Z(ff^\prime) =1+2 \left\lceil \frac{p+1}{2}  \right\rceil$.   

Besides, $f$ is not constant so by Lemma~\ref{lem_easy}, $|\Upsilon|=Z(ff^\prime)$. Hence,
\begin{align}\label{eqn_upsilon}
|\Upsilon|=Z(ff^\prime) =1+2 \left\lceil \frac{p+1}{2}  \right\rceil.
\end{align} 

We can verify that the hypothesis $|\Upsilon|\geq p$ is true. 
 Finally, equations (\ref{eqn_z(g)}) and (\ref{eqn_upsilon}) show that $Z(g)\geq(|\Upsilon|+1)(k-1)+Z(f)$.
\end{proof}

In the proof of Theorem~\ref{thm_opt}, the roots of all $W(f_1,\ldots,f_i)$ are included in the zeros of $W(f_1,\ldots,f_k)$. So, it could be possible to improve both Theorem \ref{thm_main} and Theorem \ref{thm_main3} by proving the following proposition.
\begin{open*}
  Let $f_1, \ldots, f_k$ be analytic functions on an infinite interval $I$ and $a_1, \ldots, a_k$ be non-zero real constants.
  Is the inequality 
  \begin{align*}
    Z(a_1f_1+\ldots + a_kf_k) \leq k-1 + \sum_{i=1}^kZ(W(f_1, \ldots, f_i))
  \end{align*}
always true?
\end{open*}

{\small
\section*{Acknowledgments} 

Saugata Basu pointed out to one of the authors (P.K.) 
that P\'olya and Szeg\H{o}'s book could be relevant and Maurice
Rojas pointed out Voorhoeve's papers and the fact that
Theorem~\ref{Thm_model2} implies Corollary~\ref{Cor_trinomials}. In the introduction, we mentioned the problem of bounding the number
or real solutions to an equation of the form $fg+1=0$. This question
was raised by Arkadev Chattopadhyay. 

\bibliographystyle{plain}
%\bibliography{/Users/Sebastien/Documents/Articles/Bibtex/Wronskien}
\bibliography{Wronskien}

\appendix

\section{Proof of lemmas for Section~\ref{Sec_PIT}}\label{Sec_proof_PIT}

In this section we prove two lemmas (\ref{lem_cell} and
\ref{lem_determinant}) needed for the proof of
Proposition~\ref{pro_algo}. These proofs are elementary but somewhat technical.

 In the following,
we will compute additions of $n$ integers of size $s$ in
time $O(n(s+\log n))$ and products of $n$ integers of size $s$ in
time $\tilde{O}(ns)$.

We begin by bounding the complexity of expanding a product of sparse polynomials.

\begin{lemma}\label{lem_develop}
  We consider a product of $\mu$ $\tau$-sparse polynomials $P_1, \ldots, P_{\mu}$ of degrees bounded
  by $\gamma$ with integer coefficients of size bounded by
  $\gamma$. This product can be expanded in time $\tilde{O}\left(\tau^\mu \gamma\right)$.
Moreover, the size of new coefficients is bounded by $\mu\gamma+\mu\log\tau$.
\end{lemma}

\begin{proof}
  For expanding such a product, we compute one by one each monomial of
  the sum and we store the coefficients of these $\tau^\mu$ new
  monomials. For computing one coefficient, we have three things to
  do. We have to compute its degree (sum of $\mu$ integers of size
  $\gamma$) in time $O\left(\mu(\gamma+\log \mu)\right)$, its
  coefficient (product of $\mu$ integers of size $\gamma$) in time
  $\tilde{O} \left(\mu\gamma\right)$ and we add together the monomials
  with the same exponent. At the end, at most $\tau^\mu$ coefficients
  will be added together to form a given monomial, so the size of the
  coefficient is bounded by $\mu\gamma+\mu\log\tau$. Hence, as we add
  coefficients one by one, at each step we have to add an integer of
  size $\mu\gamma$ by one of size at most
  $\mu\gamma+\mu\log\tau$. Each term of the sum takes time
\begin{align*}
  & \tilde{O}\left(\mu(\gamma+\log
    \mu)+\mu\gamma+(\mu\gamma+\mu\log\tau)\right) \\
  & = \tilde{O}\left(\mu\gamma+\mu\log(\mu\tau)\right).
\end{align*}
Therefore, computing all coefficients takes time 
\begin{align*}
  & \tilde{O}\left(\tau^\mu(\mu\gamma+\mu\log(\mu\tau))\right) \\
  & = \tilde{O}\left(\tau^\mu \gamma\right).
\end{align*}
\end{proof}

Theorem~\ref{Thm_model1} uses some constants $\beta_{\alpha,s}$
which have been defined in Lemma~\ref{lem_power}. We will need to
compute them.

\begin{lemma}\label{lem_beta}
For every $p$ in $\mathbb{N}$, we have $|\mathcal{S}_p|\leq 2^{p-1}$. For every $\alpha$, $p$ in $\mathbb{N}$ and for every $l$ in
$\mathcal{S}_p$, $0\leq\beta_{\alpha,s}\leq
(p^2+\alpha)^p$. 

Furthermore, for every $\alpha$, $p$ in $\mathbb{N}$ we can compute
all $\beta_{\alpha,s}$ with $s\in\mathcal{S}_q$ and $q\leq p$ in time $\tilde{O}
(2^p\log\alpha)$.
\end{lemma}

\begin{proof}
We showed in the proof of Lemma~\ref{lem_power} that
$\beta_{\alpha,(1,0,0,\dots)}=\alpha$ and if $s\in \mathcal{S}_p$ with
$p\neq 1$, then
\begin{multline}\label{Eq_Beta}
  \beta_{\alpha,s}=\mathbb{1}_{s_1\neq0}(\alpha-|s|+1)\beta_{\alpha,(s_1-1,s_2,s_3,\ldots)}\\
  +\underset{s_j\neq0}{\sum_{j\ :\ 2\leq j\leq  p}}(s_{j-1}+1)\beta_{\alpha,(s_1,\ldots,s_{j-1},s_{j-1}+1,s_j-1,s_{j+1},\ldots)}.
\end{multline}
However, in the formula above, the sequences $(s_1-1,s_2,s_3,\ldots)$
and $(s_1,\ldots,s_{j-1},s_{j-1}+1,s_j-1,s_{j+1},\ldots)$ fall in
$\mathcal{S}_{p-1}$. Let us denote
$M_{\alpha,p}=\underset{s\in\mathcal{S}_p}{\max}{|\beta_{\alpha,s}|}$. Hence,
if $p\neq 1$,
(\ref{Eq_Beta}) implies,
\begin{align*}
M_{\alpha,p}\leq (p^2+\alpha) M_{\alpha,p-1}.
\end{align*}
Since, $M_{\alpha,1}=\alpha$, we get by induction  $\beta_{\alpha,s}\leq
(p^2+\alpha)^{p-1}\alpha$.

For computing these constants, we notice that:
\begin{align*}
|\mathcal{S}_{p+1}| & = |\{s\in \mathcal{S}_{p+1}|s_1\neq 0\}|+|\{s\in
\mathcal{S}_{p+1}|s_1= 0\}|\\
& \leq 2|\mathcal{S}_{p}|.
\end{align*}
The inequality comes from the two surjective functions:
\begin{align*}
\mathcal{S}_{p} & \rightarrow
\{s\in \mathcal{S}_{p+1}|s_1\neq 0\} & \textrm{ and } \ \ 
\mathcal{S}_{p} & \rightarrow
\{s\in \mathcal{S}_{p+1}|s_1= 0\}\\ 
s & \mapsto (s_1+1,s_2,\ldots) & 
s & \mapsto (0,s_2,\ldots,s_{s_1},s_{1+s_1}+1,s_{2+s_1},\ldots).
\end{align*}
Hence by induction, $|\mathcal{S}_p|\leq 2^{p-1}$ and $|\bigcup_p \mathcal{S}_p|=O(2^p)$. Then, if $s\in
\mathcal{S}_p$ is fixed, for computing
$\beta_{\alpha,s}$ with (\ref{Eq_Beta}), we need to compute $p$
products of an integer of size $
p\log(p^2+\alpha)$ by an integer of size $\log p$ or $\log
\alpha$ (in time $\tilde{O}\left(p\log
(p^2+\alpha)\right)$)  and a sum of all these products in
time $\tilde{O}\left(p^2\log
(p^2+\alpha)\right)$.
Finally computing all constants $\beta_{\alpha,s}$ with
$s\in\mathcal{S}_q$ and $q\leq p$ needs time $\tilde{O}(2^pp^2\log(p^2+\alpha))$.
That proves the lemma.
\end{proof}

We can now prove the two intermediate lemmas of
Proposition~\ref{pro_algo}. For the following, we keep the notations of Proposition~\ref{pro_algo}.

\begin{lemma}\label{lem_cell}
Computing the expanded polynomial of a cell $(v,u)$ in the matrix $M$ takes time
$\tilde{O}\left(2^{vm}t^{mv}v^mC\log l\right)$. Coefficients are of
size bounded by $\tilde{O}\left(mvC\log tl\right)$ and degrees of size bounded by $Cmv$.
\end{lemma}

\begin{proof}
  We also keep the notations of Theorem~\ref{Thm_model1}.
  Each cell $(v,u)$ corresponds to the polynomial 
  \begin{multline*}
    T_{u,v}\left(
    (f_p^{(q-1)})_{1\leq p,q \leq l}\right) = \sum_{\underset{r_1+
      \ldots + r_m = v-1}{r_1,r_2, \ldots, r_m}} \prod_{j=1}^m \left[
    \sum_{s \in \mathscr{S}_{r_j}} \beta_{u,j,s} f_j^{v-1-|s|}
    \prod_{k=1}^{r_j} \left( f_j^{(k)} \right)^{s_k} \right]\\
  = \sum_{\underset{r_1+
      \ldots + r_m = v-1}{r_1,r_2, \ldots, r_m}}
  \sum_{\underset{s^i\in\mathcal{S}_{r_i}}{s^1,s^2, \ldots,
      s^m}} \left[\left(\prod_{j=1}^m\beta_{u,j,s^j}\right)\left(\prod_{j=1}^mf_j^{v-1-|s^j|}\prod_{k=1}^{r_j}\left(f_j^{(k)}\right)^{(s^j)_k}\right)\right].
\end{multline*}

The first sum is a size at most $v^m$ and the second one of size
bounded by $2^{vm}$ (Lemma~\ref{lem_beta}). For computing one term,
first, we need to compute
$\prod_{j=1}^m\beta_{u,j,s^j}=\prod_{j=1}^m\beta_{\alpha_{u,j}+l,s^j}$
wich is a product of $m$ integers, each one of size
$v\log(v^2+2^C+l)$ (since $\alpha\leq 2^C$). It is done in time
$\tilde{O}(mvC\log l)$. 

Now, we want to develop the formula with respect to $x$. We saw in the
proof of Theorem ~\ref{Thm_model1} that each monomial with respect to
the $l^2$ variables $\left(f_p^{(q-1)}\right)_{1\leq p,q\leq l}$ is of
total degree $m(v-1)$. We consider one monomial with respect to
$\left(f_p^{(q-1)}\right)_{1\leq p,q\leq l}$. It is a product of
$m(v-1)$ $t$-sparse polynomials with respect to the variable $x$.  By
Lemma~\ref{lem_develop}, this product (the second parenthesis in the
formula) can be expanded in time $\tilde{O}\left(t^{m(v-1)}C\right)$
and coefficients are of size $Cm(v-1)+m(v-1)\log t$. Then each
coefficient is first, multiplied by the corresponding coefficient
$\prod_{j=1}^m\beta_{u,j,s^j}$ in time
\begin{align*}
  & \tilde{O}\left(\max\{Cm(v-1)+m(v-1)\log
    t, mv\log(v^2+2^C+l)\}\right) \\
  & =\tilde{O}\left( mvC\log tl \right),
\end{align*}
that gives an integer of size at most $\tilde{O}\left(mvC\log tl\right)$.
Second, it is added to the stored coefficient corresponding to the same
monomial in time 
\begin{align*}
  & \tilde{O}\left(mvC\log tl+\log(v^m2^{vm})\right) \\
  & =\tilde{O}\left( mvC\log tl \right).
\end{align*}
To conclude, computing the cell takes time
\begin{align*}
  & \tilde{O}\left(v^m2^{vm}\left(mvC\log l +\left(t^{m(v-1)}C +
        t^{m(v-1)}(  mvC\log tl+mvC\log tl) \right)\right)\right)\\
  & = \tilde{O}\left(2^{vm}v^mt^{m(v-1)}C\log l)\right).
\end{align*}
That completes the proof of the lemma.
\end{proof}

\begin{lemma}\label{lem_determinant}
Assume that the entries of $M$ are given in expanded form (i.e., as
sums of monomials).
  Computing the determinant of $M$ takes time
  $\tilde{O}\left(C2^{4ml^2\log t}\right)$.
\end{lemma}

\begin{proof}
  For each one of the $(l!)$ permutations, we expand the corresponding
  polynomial. Each cell has at most $2^{ml}l^mt^{ml}$
  monomials. Hence, each permutation corresponds to a product of size
  $l$ of $\left(2^{ml}l^mt^{ml}\right)$-sparse polynomials. Powers are
  bounded by $Cml$ and coefficient sizes are bounded by
  $\tilde{O}(mlC\log t)$. By Lemma~\ref{lem_develop}, each permutation
  can be computed in time
\begin{align*}
  &\tilde{O}\left(2^{ml^2}l^{ml}t^{ml^2}mlC\log t\right)\\
  & =\tilde{O}\left(2^{3ml^2\log t}C\right)
\end{align*} and size of coefficients is bounded by 
\begin{align*}
&\tilde{O}\left(ml^2C\log t+l\log
  \left(2^{ml}l^mt^{ml}\right)\right)\\
&=\tilde{O}\left(ml^2C\log t\right).
\end{align*}

For computing the whole determinant, we compute permutations
one-by-one, adding each time new coefficients to the one computed
before. This is done in time 
\begin{align*}
  &\tilde{O}\left(l^l \left(2^{3ml^2\log t}C + 2^{ml^2}l^{ml}t^{ml^2} ml^2C\log t\right)\right)\\
  & =\tilde{O}\left(2^{4ml^2\log t}C \right).
\end{align*}
  \end{proof}

 \end{document}